\documentclass[11pt]{amsart}
\usepackage[margin=1in,bottom=1in]{geometry}
\usepackage[utf8]{inputenc}
\usepackage[english]{babel}
\usepackage{amsmath}
\usepackage{amsthm}
\usepackage{amssymb}
\usepackage{mathtools}
\usepackage{tikz-cd}
\usepackage[T1]{fontenc}
\usepackage{wasysym}
\usepackage{dsfont}
\usepackage[many]{tcolorbox}
\usepackage{soul}
\usepackage{xcolor}
\usepackage{appendix}
\usepackage{enumerate}
\usepackage{multicol}
\usepackage{fancyhdr}
\usepackage{parskip}
\usepackage{hyperref}
\usepackage{mathtools}
\usepackage{nameref}
\usepackage{thmtools}
\usepackage{thm-restate}
\usepackage{adjustbox}
\usepackage{cleveref}
\usepackage{tikz}
\usepackage{float}
\usepackage{MnSymbol}
\usepackage{amsthm}
\usepackage{orcidlink}
\usepackage{ifthen}
\usetikzlibrary{graphs.standard} 
\usepackage{tkz-euclide}
\usepackage{subfigure}
\usepackage{outlines}

\newcommand{\dist}{\mathrm{dist}}

\theoremstyle{plain} % just in case the style had changed

\newcommand{\C}{\mathcal{C}}

\newtheorem{theorem}{Theorem}[section]
\newtheorem*{theorem*}{Theorem}
\newtheorem{lemma}[theorem]{Lemma}

\newtheorem*{proposition*}{Proposition}
\newtheorem{corollary}[theorem]{Corollary}
\newtheorem*{corollary*}{Corollary}
\newtheorem*{fact*}{Fact}

\theoremstyle{definition}
\newtheorem*{definition*}{Definition}
\newtheorem{definition}[theorem]{Definition}

\newtheorem*{question*}{Question}

\newtheorem*{notation*}{Notation}

\theoremstyle{remark}

\newcommand{\Ind}[1]
{#1\setbox0=\hbox{$#1x$}\kern\wd0\hbox to 0pt{\hss$#1\mid$\hss} \lower.9\ht0\hbox to 0pt{\hss$#1\smile$\hss}\kern\wd0}

\newcommand{\notind}[1]
{#1\setbox0=\hbox{$#1x$}\kern\wd0
\hbox to 0pt{\mathchardef\nn=12854\hss$#1\nn$\kern1.4\wd0\hss}
\hbox to 0pt{\hss$#1\mid$\hss}\lower.9\ht0 \hbox to 0pt{\hss$#1\smile$\hss}\kern\wd0}

\newcommand{\N}{\mathbb{N}}

\newcommand{\Gaif}{\mathsf{Gaif}}

\title{Preservation theorems on sparse classes revisited}

\author[A. Dawar]{Anuj {Dawar}\ \orcidlink{0000-0003-4014-8248}}
\address{Anuj {Dawar}, Department of Computer Science and Technology, University of Cambridge, UK}
\email{\href{anuj.dawar@cl.cam.ac.uk}{anuj.dawar@cl.cam.ac.uk}}

\author[I. Eleftheriadis]{Ioannis {Eleftheriadis}\ \orcidlink{0000-0003-4764-8894}}
\address{Ioannis {Eleftheriadis}, Department of Computer Science and Technology, University of Cambridge, UK}
\email{\href{ie257@cam.ac.uk}{ie257@cam.ac.uk}}

\thanks{Supported by a George and Marrie Vergottis Scholarship awarded through Cambridge Trust, an Onassis Foundation Scholarship, and a Robert Sansom Studentship.}



\begin{document}

\begin{abstract}
    We revisit the work studying homomorphism preservation for first-order logic in sparse classes of structures initiated in [Atserias et al., JACM 2006] and [Dawar, JCSS 2010].  These established that first-order logic has the homomorphism preservation property in any sparse class that is monotone and addable.  It turns out that the assumption of addability is not strong enough for the proofs given.  We demonstrate this by constructing classes of graphs of bounded treewidth which are monotone and addable but fail to have homomorphism preservation.  We also show that homomorphism preservation fails on the class of planar graphs.  On the other hand, the proofs of homomorphism preservation can be recovered by replacing addability by a stronger condition of amalgamation over bottlenecks.  This is analogous to a similar condition formulated for extension preservation in [Atserias et al., SiCOMP 2008].
\end{abstract}

\maketitle

\section{Introduction}

Preservation theorems have played an important role in the development of finite model theory.  They provide a correspondence between the syntactic structure of first-order sentences and their semantic behaviour.  In the early development of finite model theory it was noted that many classical preservation theorems fail when we limit ourselves to finite structures.  An important case in point is the {\L}o{\'s}-Tarski or \emph{extension} preservation theorem, which asserts that  a first-order formula is preserved by embeddings between all structures if, and only if, it is equivalent to an existential formula. Interestingly, this was shown to fail on finite structures~\cite{Tait_1959} much before the question attracted interest in finite model theory~\cite{gurevich1984toward}.  On the other hand, the \emph{homomorphism} preservation theorem, asserting that formulas preserved by homomorphisms are precisely those equivalent to existential-positive ones, was remarkably shown to hold on finite structures by Rossman~\cite{rossman2008homomorphism}, spurring applications in constraint satisfaction and database theory.

However, even before Rossman's celebrated result, these preservation properties were investigated on subclasses of the class of finite structures.  In the case of both the extension and homomorphism preservation theorems, the direction of the theorem stating that the syntactic restriction implies the semantic closure condition is easy and holds in restriction to any class of structures.  It is the other direction that may fail and, restricting to a subclass weakens both the hypothesis and the conclusion, therefore leading to an entirely new question.  Thus, while the class of all finite structures is combinatorially wild, it contains \emph{tame} classes which are both algorithmically and model-theoretically better behaved~\cite{dawar2007finite}.  A study of preservation properties for such restricted classes of finite structures was initiated in~\cite{atseriaspreservation} and~\cite{extensionspreservation}, which looked at homomorphism preservation and extension preservation respectively.  The focus was on tame classes defined by \emph{sparsity} conditions, which allows for methods based on the \emph{locality} of first-order logic.  In particular, the sparsity conditions were based on what have come to be called wideness conditions.

We recall the formal definition of wideness in Section~\ref{sec:amalgamation} below but, informally, a class of structures $\C$ is called \emph{wide} if in any large enough structure in $\C$ we can find a large set of elements that are pairwise far away from each other.  The class $\C$ is \emph{almost wide} if there is a constant $s$ so that in any large enough structure in $\C$, removing at most $s$ elements gives a structure in which we can find a large set of elements that are pairwise far away from each other.  Finally, $\C$ is said to be \emph{quasi-wide} if  there is a function $s$ so that in any large enough structure in $\C$, removing at most $s(d)$ elements gives a structure in which we can find a large set of elements that are pairwise at distance $d$ from each other.  In the latter two cases, we refer to a set of elements whose removal yields a large scatterd set as a \emph{bottleneck} set.

The main result asserted in~\cite{atseriaspreservation} is that homomorphism preservation holds in any class $\C$ which is almost wide and is \emph{monotone} (i.e.\ closed under substructures) and \emph{addable} (i.e.\ closed under disjoint unions).  From this, it is concluded that homomorphism preservation holds for any class $\C$ whose Gaifman graphs exclude some graph $G$ as a minor, as long as $\C$ is monotone and addable.  The result was extended from almost wide to quasi-wide classes in~\cite{dawar2010homomorphism}, from which homomorphism preservation was deduced for classes that locally exclude minors and classes that have bounded expansion, again subject to the proviso that they are monotone and addable.  Quasi-wide classes were later identified with \emph{nowhere dense} classes, which are now central in structural and algorithmic graph theory \cite{nevsetvril2012sparsity}. 

The main technical construction in~\cite{atseriaspreservation} is concerned with showing that classes of graphs which exclude a minor are indeed almost wide.  The fact that homomorphism preservation holds in monotone and addable almost wide classes is deduced from a construction of Ajtai and Gurevich~\cite{ajtai} which shows the ``density'' of minimal models of a first-order sentence preserved under homomorphisms, and the fact that in an almost wide class a collection of such dense models must necessarily be finite.  While the Ajtai and Gurevich construction is carried out within the class of all finite structures, it is argued in~\cite{atseriaspreservation} that it can be carried out in any monotone and addable class because of ``the fact that disjoint union and taking a substructure are the only constructions used in the proof''~\cite[p.~216]{atseriaspreservation}.  This argument is sketched in a bit more detail in~\cite[Lemma~7]{dawar2010homomorphism}.  The starting point of the present paper is that this argument is flawed.  The construction requires us to take not just disjoint unions, but unions that identify certain elements: in other words \emph{amalgamations} over sets of points.  On the other hand, we can relax the requirement of monotonicity to just hereditariness (i.e.\ closure under induced substructures).  The conclusion is that homomorphism preservation holds in any class $\C$ that is quasi-wide, hereditary and closed under amalgamation over bottleneck points.  The precise statement is given in Theorem~\ref{thm:main} below.  We also show that the requirements formulated in~\cite{atseriaspreservation} are insufficient by constructing a class that is almost wide (indeed, has bounded treewidth), is monotone and addable, but fails to have the homomorphism preservation property.

Interestingly, the requirement of amalgamations over bottlenecks is similar to that used to define classes on which the extension preservation property holds  in~\cite{extensionspreservation}, even though the construction uses rather different methods.  The result there can be understood, in our terms, as showing that the extension preservation theorem holds in any almost wide, hereditary class with amalgamation over bottlenecks.  As we observe below (in Corollary~\ref{cor:implies}), this implies that homomorphism preservation holds in all such classes.  Our Theorem~\ref{thm:main} then strengthens this to quasi-wide classes, where we do not know if extension preservation holds.  The class of planar graphs is an interesting case as it is used in~\cite{extensionspreservation} as an example of a hereditary, addable class with excluded minors in which extension preservation fails.  We show here that homomorphism preservation also fails in this class, strengthening the result of~\cite{extensionspreservation}.

In the rest of this paper, we introduce notation and necessary background in Section~\ref{sec:prelim}.  We construct a monotone, addable class of graphs of small treewidth in Section~\ref{sec:treewidth}, providing the first counterexample to the claims of~\cite{atseriaspreservation}.  Section~\ref{sec:amalgamation} states and proves the corrected version of the preservation theorems and Section~\ref{sec:planar} shows the failure of homomorphism preservation in the class of planar graphs.

\section{Preliminaries}\label{sec:prelim}

We assume familiarity with the standard notions of finite model theory and structural graph theory, and refer to \cite{ebbinghaus} and \cite{nevsetvril2012sparsity} for reference. We henceforth fix a finite relational vocabulary $\tau$; by a structure we implicitly mean a $\tau$-structure. 
We often abuse notation and do not distinguish between a structure and its domain.
Given two structures $A,B$, a homomorphism $f:A \to B$ is a map such that for all relation symbols $R$ and tuples $\bar a$ from $A$ we have $\bar a \in R^A \implies f(\bar a) \in R^B$. If moreover $f(\bar a) \in R^B \implies \bar a \in R^A$ then $f$ is said to be \emph{strong}. An injective strong homomorphism is called an \emph{embedding}. We also call a homomorphism $f:A \to B$  \emph{full} if it is surjective and for any relation symbol $R$ and tuple $\bar b$ from $B$ we have $\bar b \in R^B \implies \exists \bar a \in R^A \text{ with }f(\bar a)=\bar b$.

A structure $B$ is said to be a \emph{weak substructure} 
%(or \emph{subgraph} in the case of graphs)
of a structure $A$ if $B \subseteq A$ and the inclusion map $\iota: B \hookrightarrow A$ is a homomorphism. Likewise, $B$ is an \emph{induced substructure} 
%(or \emph{induced subgraph} in the case of graphs) 
of $A$ if the inclusion map is an embedding. 
Given a structure $A$ and a subset $S \subseteq A$ we write $A[S]$ for the unique induced substructure of $A$ with underlying set $S$. 
An induced substructure $B$ of $A$ is said to be \emph{free in $A$} if there is some structure $C$ such that $A$ is the disjoint union $B + C$. Finally, a substructure $B$ of $A$ is said to be \emph{proper} if the inclusion map is not full. We say that a class of structures is \emph{monotone} if it is closed under weak substructures, and it is \emph{hereditary} if it is closed under induced substructures. Moreover a class is called \emph{addable} if it is closed under taking disjoint unions. We often consider classes of undirected graphs.  Seen as a relational structure, this is a set with an irreflexive symmetric relation $E$ on it.  A weak substructure of such a graph need not be a graph.  However, when we speak of a monotone class of undirected graphs, we mean it in the usual sense of a class of graphs closed under the operations of removing edges and vertices.

Given a structure $A$ and an equivalence relation $E \subseteq A \times A$ we define the quotient structure $A / E$ as the structure whose domain $A / E = \{[a]_E: a \in A\}$ is the set of $E$-equivalence classes and such that for every relation symbol $R$ of arity $n$ we have $R^{A/E}=\{([a_1],\dots,[a_n]) \in A/E: (a_1,\dots,a_n) \in R^A\}$. The quotient map $\pi_E: A \twoheadrightarrow A/E$ is a full homomorphism which we call the \emph{quotient homomorphism}. Given structures $A,B$ and a set $S\subseteq A\cap B$ such that $A[S]=B[S]$, we write $A \oplus_S B$ for the \emph{free amalgam of $A$ and $B$ over $S$}. This can be defined as the quotient of the disjoint union $A+B$ by the equivalence relation generated by $\{(\iota_A(s),\iota_B(s)):s \in S\}$, where $\iota_A:S \to A$ and $\iota_B:S \to B$ are the inclusion maps. Evidently, there is an injective homomorphism $A \to A\oplus_S B$ given by composing the inclusion $A \to A+B$ with the quotient $A+B \to A\oplus_S B$, and a full homomorphism $A+A \to A$ which descends to a full homomorphism $A\oplus_S A \to A$. 

%Recall that the chromatic number of an undirected graph $G$, denoted by $\chi(G)$, is the least $n \in \N$ such that there is a homomorphism from $G$ to the complete graph $K_n$ on $n$ vertices. Clearly, if there is a homomorphism $G \to H$ then $\chi(G) \leq \chi(H)$. A graph $H$ is a \emph{minor} of a graph $G$ if there is a map $f:H \to \pow(G)$ such that for all $u,v \in H$, $f(u)$ is non-empty and connected, $f(u)\cap f(v) = \emptyset$, and if there is an edge $(u,v) \in E(H)$ then there are $a \in f(u)$ and $b \in f(v)$ such that $(a,b) \in E(G)$. 

By the \emph{Gaifman graph} of a structure $A$ we mean the undirected graph $\Gaif(A)$ with vertex set $A$  such that two elements are adjacent if, and only if, they appear together in some tuple of a relation of $A$. Given a structure $A$, $r \in \N$, and $a \in A$, we write $B^r_A(a)$ for the \emph{ball of radius $r$ around $a$ in $A$}, that is, the set of elements of $M$ whose distance from $a$ in $\Gaif(A)$ is at most $r$; we shall often abuse notation and write $B^r_A(a)$ to mean the induced substructure $A[B^r_A(a)]$ of $A$, possibly with a constant for the element $a$. A set $S \subseteq A$ is said to be \emph{$r$-independent} if $b \notin B^r_A(a)$ for any $a,b \in A$. 

For $r \in \N$, let $\mathrm{dist}(x,y)\leq r$ be the first-order formula expressing that the distance between $x$ and $y$ in the Gaifman graph is at most $r$, and $\mathrm{dist}(x,y)> r$ its negation. Clearly, the quantifier rank of $\mathrm{dist}(x,y)\leq r$ less than $r$. A \emph{basic local sentence} is a sentence
\[ \exists x_1, \dots, x_n (\bigwedge_{i \neq j} \mathrm{dist}(x_i,x_j)> 2r \land \bigwedge_{i \in [n]} \psi^{B^r(x_i)}(x_i)),\]
where $\psi^{B^r(x_i)}(x_i)$ denotes the relativisation of $\psi$ to the $r$-ball around $x_i$, i.e. the formula obtained from $\psi$ by replacing every quantifier $\exists x \ \theta$ with $\exists x (\dist(x_i,x)\leq r \land \theta)$, and likewise every quantifier $\forall x \ \theta$ with $\forall x (\dist(x_i,x)\leq r \to \theta)$. We call $r$ the \emph{locality radius}, $n$ the \emph{width}, and $\psi$ the \emph{local condition} of $\phi$.  Recall the Gaifman locality theorem~\cite[Theorem~2.5.1]{ebbinghaus}.

\begin{theorem}[Gaifman Locality]
    Every first-order sentence of quantifier rank $q$ is equivalent to a Boolean combination of basic local sentences of locality radius $7^q$. 
\end{theorem}

We say that a formula $\phi$ is preserved by homomorphisms (resp. extensions) over a class of structures $\C$ if for all $A,B \in \C$ such that there is a homomorphism (resp. embedding) from $A$ to $B$, $A \models \phi$ implies that $B \models \phi$. We say that a class of structures $\C$ has the \emph{homomorphism preservation property} (resp. \emph{extension preservation property}) if for every formula $\phi$ preserved by homomorphisms (resp. extensions) over $\C$ there is an existential-positive (resp. existential) formula $\psi$ such that $M \models \phi \iff M \models \psi$ for all $M \in \C$. 

Given a formula $\phi$ and a class of structures $\C$, we say that $M \in \C$ is a \emph{minimal induced model} of $\phi$ in $\C$ if $M\models \phi$ and for any proper induced substructure $N$ of $M$ with $N \in \C$ we have $N \not\models \phi$. The relationship between minimal models and preservation theorems is highlighted by the following lemma, which is standard, and combines~\cite[Theorem~3.1]{atseriaspreservation} and~\cite[Lemma~2.1]{extensionspreservation}

\begin{lemma}\label{lem:minimalmodels}
    Let $\C$ be a hereditary class of finite structures. Then a sentence preserved under homomorphisms (resp. extensions) in $\C$ is equivalent to an existential-positive (resp. existential) sentence over $\C$ if and only if it has finitely many minimal induced models in $\C$.
\end{lemma}

\begin{proof}
    We provide a proof for preservation under homomorphisms, as proof for the other case is very similar. Suppose that $\phi$ has finitely many minimal induced models in $\C$, say $M_1,\dots,M_n$. Let $\psi_i$ be the canonical query of each $M_i$, and write $\psi:= \bigvee_{i \in [n]} \psi_i$. We argue that $\phi$ is equivalent to $\psi$ over $\C$. Indeed, if $A \in \C$ models $\phi$ then $A$ contains a minimal induced model $B$ of $\phi$ as an induced substructure. By hereditariness $B \in \C$. Hence, $B$ is isomorphic to some $M_i$. Since there is clearly a homomorphism $B \to A$ it follows that $A \models \psi$. On the other hand if $A \models \psi$, then some $M_i$ has a homomorphism to $A$. Since $M_i \models \psi$ and $\psi$ is preserved by homomorphisms this implies that $A \models \phi$ as required.

    Conversely, assume that $\phi$ is equivalent to an existential positive sentence over $\C$. In particular, $\phi$ is equivalent to some disjunction $\bigvee_{i \in [n]} \psi_i$ where each $\psi_i$ is primitive positive. Let $M_i$ be the canonical database of $\psi_i$. Now, if $A$ is a minimal induced model of $\phi$ then in particular $A \models \psi_i$ for some $i \in [n]$, i.e. there is a homomorphism $h:M_i \to A$. If $h$ is not surjective, then $A[h[M_i]]$ is a proper induced substructure of $A$, which is in $\C$ by hereditariness, and models $\phi$; this contradicts the minimality of $A$. Hence, the size of every minimal induced model of $\phi$ in $\C$ is bounded by $\max_{i \in [n]} |M_i|$. It follows that $\phi$ can have only finitely many minimal induced models in $\C$. 
\end{proof}

\begin{corollary}\label{cor:implies}
Let $\C$ be a hereditary class of finite structures. If $\C$ has the extension preservation property, then $\C$ has the homomorphism preservation property. 
\end{corollary}

\begin{proof}
    If a formula $\phi$ is preserved by homomorphisms then, in particular, it is preserved by extensions. It follows that $\phi$ is equivalent to an existential sentence over $\C$, and so by \Cref{lem:minimalmodels} it has finitely many minimal induced models in $\C$. Consequently, the same lemma implies that $\phi$ is equivalent to an existential-positive sentence over $\C$ as required. 
\end{proof}

Another immediate consequence of the above is that both preservation properties hold over any class $\C$ that is \emph{well-quasi-ordered} by the induced substructure relation, i.e. for every infinite subclass $\{M_i : i \in I\} \subseteq \C$ there are $i\neq j \in I$ such that either $M_i$ is an induced substructure of $M_j$ or vice versa. In fact, any property (i.e. not necessarily definable by a first-order formula) preserved by homomorphisms (resp. extensions) is equivalent to an existential positive (resp. existential) formula over a well-quasi-ordered class. In particular, this applies to classes of cliques or more generally classes of bounded shrub-depth \cite{ganian2019shrub}.

\section{Preservation can fail on classes of small treewidth}\label{sec:treewidth}

Theorem~4.4 of~\cite{atseriaspreservation} can be paraphrased in the language of this paper as saying that \emph{homomorphism preservation holds over any monotone and addable class of bounded treewidth}.
In this section we provide a simple counterexample to this, exhibiting a monotone and addable class of graphs of treewidth $3$ where homomorphism preservation fails. More generally, this contradicts Corollary~3.3 of \cite{atseriaspreservation} and Theorem~9 of \cite{dawar2010homomorphism}. To witness failure of preservation, we must exhibit the relevant class, a formula preserved by homomorphisms over this class, and an infinite collection of minimal induced models in the class. We then conclude by \Cref{lem:minimalmodels}.

\begin{definition}
     Fix $k \in \N$ and $n_i \geq 3$ for every $i \in [k]$. We define the \emph{bouquet of cycles of type $(n_1,\dots,n_k)$}, denoted by $W_{n_1,\dots,n_k}$, as the graph obtained by taking the disjoint union of $k$ cycles of length $n_1,\dots,n_k$ respectively, and adding an apex vertex, i.e. a vertex adjacent to every vertex in these cycles. Whenever $k = 1$, we refer to the graph $W_n$ as the \emph{wheel of order $n$}.
\end{definition}

\begin{figure}[h!]
    \centering
    \begin{tikzpicture}
  % Nodes
  
  \node[circle, draw, fill=black, inner sep=0.8pt] (P) at (0,2) {};
  
    \begin{scope}[shift={(-3,0)}]
    \foreach \i in {1,...,6} {
    \node[circle, draw, fill=black, inner sep=0.8pt] (A\i) at ({360/6 * (\i - 1)}:1) {};
  }
    \foreach \i in {1,...,6} {
    \draw (A\i.center) -- (P.center);
    \edef\next{\number\numexpr \i+1}%
    \ifnum \i<6
      \draw (A\i.center) -- (A\next.center);
      \fi
    }
    \draw (A1.center) -- (A6.center);
  \end{scope}

    \begin{scope}[shift={(0,0)}]
    \foreach \i in {1,...,9} {
    \node[circle, draw, fill=black, inner sep=0.8pt] (B\i) at ({360/9 * (\i - 1)}:1) {};
  }
    \foreach \i in {1,...,9} {
    \draw (B\i.center) -- (P.center);
    \edef\next{\number\numexpr \i+1}%
    \ifnum \i<9
      \draw (B\i.center) -- (B\next.center);
      \fi
    }
    \draw (B1.center) -- (B9.center);
\end{scope}
  \begin{scope}[shift={(3,0)}]
    \foreach \i in {1,...,10} {
    \node[circle, draw, fill=black, inner sep=0.8pt] (C\i) at ({360/10 * (\i - 1)}:1) {};
  }
    \foreach \i in {1,...,10} {
    \draw (C\i.center) -- (P.center);
    \edef\next{\number\numexpr \i+1}%
    \ifnum \i<10
      \draw (C\i.center) -- (C\next.center);
      \fi
    }
    \draw (C1.center) -- (C10.center);
  \end{scope}
\end{tikzpicture}
\begin{tikzpicture}[scale=0.7]
  % Nodes
  \foreach \i in {1,...,9} {
    \node[circle, draw, fill=black, inner sep=0.8pt] (N\i) at ({360/9 * (\i - 1)}:2) {};
  }
  \node[circle, draw, fill=black, inner sep=0.8pt] (C) at (0,0) {};
  \node[circle, draw=none] (P) at (-3,0) {};
  % Edges (original)
  \foreach \i in {1,...,9} {
    \draw (N\i.center) -- (C.center);
    \edef\next{\number\numexpr \i+1}%
    \ifnum \i<9
      \draw (N\i.center) -- (N\next.center);
      \fi
    }
    \draw (N1.center) -- (N9.center);

\end{tikzpicture}
    \caption{The bouquet of cycles of type $(6,9,10)$ and the wheel of order $9$ respectively.}
    \label{fig:wheelbouquet}
\end{figure}
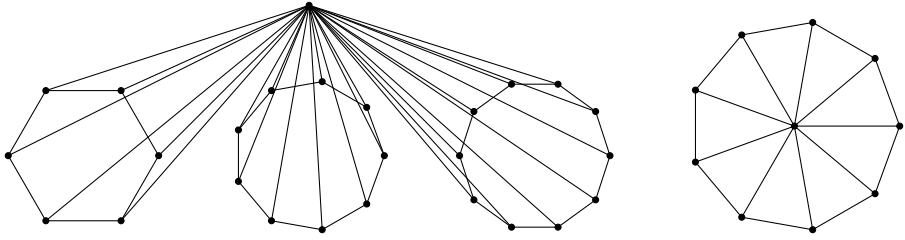

\begin{lemma}\label{lem:full}
Fix $n,m \in N$ odd. Then the wheel $W_n$ has chromatic number $4$, while any proper subgraph of $W_n$ has chromatic number $3$. Consequently, any homomorphism $f:W_n \to W_m$ is full. 
\end{lemma}

\begin{proof}
    Fix $n,m \in \N$ odd. It is clear that any proper colouring of $W_n$ must use a unique colour for the apex as it is adjacent to every other node in $W_n$. Moreover, we require an additional three colours to colour the vertices in the odd-length cycle of $W_n$. It follows that $\chi(W_n)=4$. Now, let $W$ be a proper subgraph of $W_n$. It follows that there is at least one edge $(u,v)$ present in $W_n$ which is not in $W$. If $u$ and $v$ are both in the cycle of odd length then we may define a proper $3$-colouring of $W$ by giving $u$ and $v$ the same colour, alternating between this and a second colour along the cycle, and using a final third colour for the apex. If one of $u$ or $v$ is the apex of $W_n$, then we once again colour $u$ and $v$ with the same colour and use an additional two colours to alternate between along the cycle. In particular, it follows that the chromatic number of any homomorphic image of $W_n$ is at least $4$, and so it cannot be a proper substructure of $W_m$. This implies that any homomorphism $f:W_n \to W_m$ is full. 
\end{proof}

The advantage of working with bouquets of cycles is that, unlike single cycles, there is a formula that defines their existence as free induced subgraphs. To see this, we first let
\[ \psi(x,z):=\exists u \exists v [u\neq v \land u \neq x \land v \neq x \land E(z,u) \land E(z,v) \land\forall w( E(w,z) \to w = u \lor w=v \lor w=x)],\]
with which we define 
\[ \phi:= \exists x \exists y [E(x,y) \land \forall z (z \neq x \land \dist(x,z)\leq 2 \to E(x,z) \land \psi(x,z)].\]
Intuitively, $\phi$ asserts the following: ``there is a vertex $x$ of degree at least one such that every other vertex reachable from $x$ by a path of length two is adjacent to $x$ and has exactly two distinct neighbours which are not $x$''.  

\begin{lemma}\label{lem:phi}
    Let $G$ be an arbitrary finite graph. Then $G \models \phi$ if, and only if, it contains a bouquet of cycles as a free induced subgraph. 
\end{lemma}

\begin{proof}
    Suppose that $G$ contains a bouquet $W$ of cycles as a free induced subgraph. Then the apex of the bouquet is a vertex of degree at least one, while every vertex reachable from the apex by a path of length two must be in one of the cycles, since $W$ is free in $G$. Since all vertices in the cycles have degree exactly two, not considering the apex itself, it follows that $G\models \phi$.

    Conversely, suppose that $G \models \phi$, and let $x$ be the vertex that is guaranteed to exist by $\phi$. Let $S \subseteq V(G)$ be the vertices that are adjacent to $x$. Since $x$ has degree at least one it follows that $S$ is non-empty. Partition $S$ into $k$ classes $S_1,\dots,S_k$, by putting two vertices in the same class if, and only if, there is a path between them in $G \setminus \{x\}$. We argue that for each $i \in [k]$, $S_i$ is a free induced cycle in $G \setminus \{x\}$. First, notice that since every vertex reachable from $x$ by a path of length two has degree exactly two in $G \setminus \{x\}$, it follows that $S_i$ induces a cycle in $G \setminus \{x\}$. Moreover, there is no $y \in G \setminus (\{x\}\cup S_i)$ which is adjacent to $S_i$. Indeed, if $y \neq x$ is adjacent to some $v \in S_i$ then $y$ is reachable from $x$ by a path of length two. It follows that $y$ is itself adjacent to $x$, and therefore $y$ is in $S$; in particular, $y$ and $v$ are in the same class and so $y \in S_i$. Consequently each $S_i$ is a free induced cycle in $G \setminus \{x\}$, and so the connected component of $x$ is a free induced bouquet of cycles in $G$. 
\end{proof}

We note that going beyond finite graphs to infinite graphs our formula $\phi$ no longer defines bouquets of cycles, as witnessed by a path of infinite length and an apex vertex. Moreover, it is evident that $\phi$ is not preserved by homomorphisms in general as every $W_n$ maps homomorphically to the structure $W_n \cup \{c\}$ with an additional vertex adjacent to the apex, and while the latter contains a bouquet of cycles as an induced subgraph, it does not contain a bouquet of cycles as a free induced subgraph. However, when restricting to subgraphs of disjoint unions of wheels we no longer have non-free-occurring bouquets of cycles in the class. This is precisely the core of the following argument. 

\begin{theorem}\label{thm:wheels}
    Let $\C$ be the closure of $\{W_{2n+1} :n \in \N\}$ under taking subgraphs and disjoint unions. Then homomorphism preservation fails in $\C$. 
\end{theorem}

\begin{proof}
    Let $\phi$ be as above. By \Cref{lem:phi} it follows that every $W_n$ is a model of $\phi$. Clearly, proper subgraphs of $W_n$ cannot possibly contain a bouquet of cycles as a free induced subgraph, and so they cannot model $\phi$ by \Cref{lem:phi}. Consequently, each $W_n$ is a minimal model of $\phi$ in $\C$. 
    
    We now argue that $\phi$ is preserved by homomorphisms in $\C$. Indeed, if some $G \in \C$ is such that $G \models \phi$ then by \Cref{lem:phi} it contains a bouquet of cycles as a free induced subgraph. By the choice of $\C$, this necessarily implies that $G$ contains $W_n$ as a free induced subgraph for some odd $n$. Let $H \in \C$ and $f:G \to H$ be a homomorphism. Then $f$ restricts onto a homomorphism $W_n \to H$ which, by the connectivity of $W_n$ and the fact that $H \in \C$, descends to a homomorphism $\hat{f}:W_n \to W_m$ for some odd $m \in \N$. It follows by \Cref{lem:full} that $\hat{f}$ is full, and therefore $H$ contains $W_m$ as a subgraph. The choice of $\C$ once again ensures that $W_m$ is a free induced subgraph of $H$, and so \Cref{lem:phi} implies that $H \models \phi$ as required.

    We finally conclude by \Cref{lem:minimalmodels} that $\phi$ is not equivalent to an existential-positive formula over $\C$ since it has infinitely many minimal models in $\C$.
\end{proof}

Finally, observe that each graph $W_n$ has treewidth $3$ (in fact even pathwidth $3$). Indeed, taking a tree decomposition of the cycle $C_n$ of width $2$, and adding the apex to every bag in the decomposition gives the required tree decomposition of $W_n$. 

\section{Preservation under bottleneck amalgamation}\label{sec:amalgamation}

The results of \cite{atseriaspreservation} were later extended to classes that are \emph{quasi-wide}. Recall that a class is called quasi-wide if for every $r \in \N$ there exist $k_r \in \N$ and $f_r : \N \to \N$ such that for all $m \in \N$ and $M \in \C$ of size at least $f_r(m)$ there are disjoint sets $A,S \subseteq M$ with $|A|\geq m$, $|S|\leq k_r$ and such that $A$ is $r$-independent in $M \setminus S$. Intuitively, quasi-wideness ensures that, for every choice of distance, we may find in suitably large structures a large set of elements that are pairwise far away after removing a small set of \emph{bottleneck points}. This notion was precisely introduced for the purpose of extending the arguments of \cite{atseriaspreservation} to more general sparse classes, such as classes of bounded expansion, or classes that locally exclude a minor. In particular, Theorem~9 of~\cite{dawar2010homomorphism} asserts (paraphrased into the language of this paper) that: \emph{homomorphism preservation holds over any monotone and addable quasi-wide class}.
Evidently, this is violated by \Cref{thm:wheels} above. Nonetheless, we may salvage the proof by replacing additivity by the stronger assumption of closure under amalgamation over the bottleneck points that witness quasi-wideness.  

%Quasi-wide classes where introduced in \cite{dawar2010homomorphism} to precisely generalise (the now refuted) result of \cite{atseriaspreservation} that held for \emph{almost-wide} classes, i.e. those where the size of the bottleneck points does not depend on the radius of independence. 

The proof proceeds by arguing that any suitably large model $M$ of a sentence $\phi$ preserved by homomorphisms over a class $\C$ satisfying our assumptions, has a proper induced substructure $N$ which also models $\phi$. We thus obtain a concrete bound on the size of minimal models of $\phi$, and conclude by \Cref{lem:minimalmodels}. The existence of this bound is guaranteed by quasi-wideness, as any large enough structure contains a large scattered set after removing a small number of bottleneck points. To isolate the bottleneck points $\bar p$ of $M$ we consider a structure $\bar p M$ in an expanded language which is bi-interpretable with $M$, and work with the corresponding interpretation $\phi^k$ of $\phi$; in particular $\bar p M$ contains a large scattered set itself and it models $\phi^k$. Then, by removing a carefully chosen point from the scattered set of $\bar p M$, we obtain a proper induced substructure $\bar p N $ of $ \bar pM$ such that $N \in \C$ by hereditariness. To argue that this still models $\phi^k$, we use a relativisation of the locality argument of Ajtai and Gurevich from \cite{ajtai}. While in its original version the argument only considers disjoint copies of $M$, working with the interpretation $\bar p M$ of $M$ corresponds to taking free amalgams of $M$ over the set of bottleneck points; this is precisely the subtlety that was missed in \cite{atseriaspreservation} and \cite{dawar2010homomorphism}. 

We now define the structure $\bar p M$; in the following we only consider the case of undirected graphs for simplicity. For arbitrary relational structures the idea is analogous, in that we isolate the tuple $\bar p$ by forgetting any relation that contains some $p_i$, and introduce new relation symbols of smaller arities to recover the forgotten relations. 

\begin{definition}
    Fix $k \in \N$, and let $\sigma=\{E,P_1,\dots,P_k,Q_1,\dots,Q_k\}$ be the expansion of the language of graphs with $2k$ unary predicates. Given a graph $G=(V,E)$ and a tuple $(p_1,p_2,\dots,p_k) \in V^k$, define the $\sigma$-structure $\bar p G$ on the same domain $V$ such that for all $i \in [k]:$
    \begin{itemize}
        \item $E^{\bar p G} = \{(u,v) \in E: u,v \notin \{p_1,\dots,p_k\}\}$;
        \item $P_i^{\bar p G} = \{p_i\}$;
        \item $Q_i^{\bar p G} = \{v \in V: (p_i,v) \in E\}.$

    \end{itemize}
    
    Consider the formula $\epsilon(x,y):= \bigvee_{i \in [k]} (P_i(x) \land Q_i(y)) \lor E(x,y)$. Given a sentence $\phi$, write $\phi^k$ for the $\sigma$-sentence obtained by $\phi$ by replacing every atom $E(x,y)$ by $(\epsilon(x,y) \lor \epsilon(y,x))$. It is then clear that for every $G=(V,E)$ and $\bar p \in V^k$ as above:
    \[ G \models \phi \iff \bar pG \models \phi^k.\]
\end{definition}

With this, we turn to our main theorem in this section. Instead of proving a base case and invoking that for the interpretation step as in \cite{extensionspreservation}, \cite{atseriaspreservation}, and \cite{dawar2010homomorphism}, we opt for a direct proof to illustrate the relevance of our assumptions on the class. 

\begin{theorem}\label{thm:main}
Let $\C$ be a hereditary class such that for every $r \in \N$ there exist $k_r \in \N$ and $f_r:\N\to \N$ so that for every $m \in \N$ and $M \in \C$ of size at least $f_r(m)$ there exist disjoint sets $A,S \subseteq M$ with:
\begin{itemize}
    \item $|A|\geq m$ and $|S|\leq k_r$;
    \item $A$ being $r$-independent in $M \setminus S$; 
    \item $\oplus^n_S M:= \underbrace{M \oplus_S M \oplus_S \dots \oplus_S M}_\text{$n$ times} \in \C$ for every $n \in \N$.
\end{itemize}

Then homomorphism preservation holds over $\C$.

\end{theorem}

 \begin{proof}
 Let $\C$ be as above, and fix $\phi$ which is preserved by homomorphisms over $\C$. Denote the quantifier rank of $\phi$ by $q$, and let $r=2\cdot 7^q$. It follows that there is some $k \in \N$ and some $f:\N \to \N$ such that for every $m \in \N$ and every $M$ in $\C$ of size at least $f(m)$, there are disjoint sets $A,S\subseteq M$ such that $|A|\geq m$, $|S|\leq k$, $|A|$ is $r$-independent in $M\setminus S$, and $\oplus_S ^n M \in \C$ for every $n \in \N$.  We consider the formula $\phi^k$: by Gaifman locality, there is a set $\{\phi_1,\dots,\phi_s\}$ of basic local sentences such that $\phi^k$ is equivalent to a Boolean combination of these. For $i \in [s]$, let $r_i$ and $n_i$ be the radius and width of locality respectively of $\phi_i$, and $\psi_i$ its local condition. Observe that $\phi$ and $\phi^k$ have the same quantifier rank, and so $2\cdot \max_{i \in [s]} r_i \leq r$. Set $n:= \max_{i \in [s]} n_i$ and $m:=2^s+1$. We argue that every minimal model $M \in \C$ of $\phi$ has size $<f(m)$.

 Let $M$ be a minimal model of $\phi$, and assume for a contradiction that $|M|\geq f(m)$. It follows that there is a set $S \subseteq M$ of size $k$ such that $M \setminus S$ contains an $r$-independent set $A$ of size $m$. Let $\bar p \in M^k$ be an enumeration of $S$; this implies that $\bar p M \models \phi^k$ and $A$ is an $r$-independent set in $\bar p M$. For each $i \in [s]$ define 
    \[ \Psi_i(x) := \exists y (\dist(x,y)\leq r_i \land \psi_i^{B^{r_i}(y)}(y)). \]
 Since $|A|\geq m= 2^s+1$, it follows that there are at least two vertices $u,v \in A$ satisfying 
 \[ B^{2r_i}_{\bar p M}(u) \models \Psi_i(u) \iff B^{2r_i}_{\bar p M}(v) \models \Psi_i(v) \]
for all $i \in [s]$. Let $N'$ be the substructure of $\bar p M$ induced on $\bar p M \setminus \{u\}$. Since $A$ does not intersect the vertices in $S$, the substructure $N$ of $M$ induced on $M \setminus \{u\}$ satisfies that $N'=\bar p N$. We shall argue that $\bar p N \models \phi^k$ and so $N \models \phi$, contradicting that $M$ is a minimal induced model of $\phi$.
 
By our closure assumptions on $\C$, $N_n:=\oplus_{S}^n N$ and $M_n:=M\oplus_{S} (\oplus_{S}^n N)$ are both in $\C$, as the latter is an induced substructure of $\oplus_S^{n+1} M$. Since there is a homomorphism $M \to M_n$ we obtain that $M_n \models \phi$ and thus $\bar p(M_n)\models \phi^k$. We shall argue that
\[ \bar p(M_n)\models \phi_i \iff \bar p(N_n)\models \phi_i\]
for all basic local sentences $\phi_i$ of $\phi^k$. In particular, this implies that $\bar p(N_n)\models \phi^k$, and so $N_n \models \phi$. Since there is a homomorphism $N_n \to N$, the preservation of $\phi$ implies that $N \models \phi$ as claimed.

Clearly, if $\bar p (N_n) \models \phi_i$ then $\bar p(M_n)\models\phi_i$ by the fact that $\phi_i$ is a local sentence and $\bar p(N_n)$ a free induced substructure of $\bar p(M_n)$. Conversely, if $\bar p(M_n) \models \phi_i$ then there is a $2r_i$-independent subset $X$ of size $n_i$ such that $B^{r_i}_{\bar p(M_n)}(x) \models \psi_i(x)$ for every $x \in X$. Observe that if $X \subseteq S$ then clearly $\bar p(N_n)\models \phi_i$. So, since $S$ is isolated in $\bar p (M_n)$ and $\bar p (N_n)$, we focus on elements of $X \setminus S$, which we may assume to be non-empty. We therefore distinguish two cases. 

    If $|X \setminus S| > 1$ then, by the $2r_i$-independence of $X$, there is at least one $x \in X\setminus S$ such that $u \notin B_{\bar p(M_n)}^{r_i}(x)$. It follows that the $r_i$-ball centered at $x$ is isomorphic to an $r_i$-ball centered at an element in a disjoint copy of $N\setminus S$. Since $\bar p(N_n)$ contains $n\geq n_i$ such copies, it follows that there is a $2r_i$-independent subset $Y$ of $N_n$ of size $n_i$ such that $B^{r_i}_{\bar p(N_n)}(y) \models \psi_i(y)$ for every $y \in Y$, i.e. $\bar p(N_n) \models \phi_i$.
    
    On the other hand if $|X \setminus S| = 1$, let $x$ be the unique element of $X \setminus S$. Clearly if $u \notin B^{r_i}_{\bar p(M_n)}(x)$ then the $r_i$-ball centered at $x$ is isomorphic to the $r_i$-ball centered at an element in $\bar p(N_n)$, and so $\bar p(N_n) \models \phi_i$. If $u \in B^{r_i}_{\bar p(M_n)}(x)$, then $\bar p M \models \Psi_i(u)$, and so by the choice of $u$ and $v$, $\bar p M \models \Psi_i(v)$. Consequently, there is some $y \in B^{r_i}_{\bar p M}(v)$ such that $B^{r_i}_{\bar p M}(v) \models \psi_i(v)$. Observe that because $v$ and $u$ are $2r_i$-independent, $u \notin B_{\bar p(M_n)}^{r_i}(y)$. As before, this implies that $\bar p(B_n) \models \phi_i$. 

    The above implies that there are finitely many minimal induced models of $\phi$ in $\C$, and so we conclude that  $\phi$ is equivalent to an existential-positive formula over $\C$ by \Cref{lem:minimalmodels}. 
 \end{proof}

% \begin{corollary}\label{cor:mainth}
%Let $\C$ be a hereditary quasi-wide class with the free amalgamation property. Then $\C$ has the homomorphism preservation property. 
 %\end{corollary}

 Going back to bouquets of cycles, it is easy to see that if a bouquet has more than $m^2\cdot (r+1)$ vertices then after removing the apex it either contains $m$ disjoint cycles or it contains a cycle of size $m\cdot(r+1)$; in either case it contains an $r$-independent set of size $m$. In this case the apex is the only bottleneck point, and so amalgamating over this corresponds to adding more cycles to the bouquet. Consequently, homomorphism preservation holds for the hereditary closure of the class of bouquets of cycles by \Cref{thm:main} above. 
 
   Closure under amalgamation over bottlenecks is a technical condition and one might consider if it could be replaced by more natural conditions.  For example, we could strengthen it by considering closure under arbitrary amalgamation.  However, this is a condition that does not sit well with sparsity requirements.  Indeed, any hereditary class of undirected graphs that is also closed under arbitrary amalgamation contains arbitrarily large $1$-subdivided cliques, and hence, cannot be quasi-wide.
 Nonetheless, there are naturally defined sparse families of structures that satisfy the conditions of Theorem~\ref{thm:main}.   One such class is known to exist by \cite{extensionspreservation}, that is, the class $\mathcal{T}_k$ of all graphs of treewidth bounded by $k$, for any value of $k \in \N$. Indeed, for any suitably large graph of bounded treewidth we may pick a set of bottleneck points that comes from the same bag in a tree decomposition of the graph, and so amalgamating over this set of points does not increase the treewidth. Another naturally defined such class is the class of outerplanar graphs. For our purposes, we may define outerplanar graphs as those omitting $K_4$ and $K_{2,3}$ as minors \cite{outerplanar}. 
 The quasi-wideness of this class follows by the next fact, which moreover permits some control over the bottleneck points. 
 
\begin{theorem}\cite{atseriaspreservation}\label{fact:cliqueminor}
    For every $k,r,m \in \N$ there is an $N=N(k,r,m) \in \N$ such that if $G$ is a graph of size at least $N$ excluding $K_k$ as a minor, then there are disjoint sets $A,S \subseteq V(G)$ with $|A|\geq m$ and $|S| \leq k-2$ such that $A$ is $2r$-independent in $G\setminus S$. Moreover, the bipartite graph $K_{A,S}$ with parts $A$ and $S$ defined by putting an edge between $a \in A$ and $s \in S$ if and only if there is some $u \in B^r_{G\setminus S}(a)$ such that $(u,s) \in E(G)$ is complete. 
\end{theorem}

 \begin{theorem}
     Homomorphism preservation holds for the class of outerplanar graphs. 
 \end{theorem}

  \begin{proof}
     Since outerplanar graphs are $K_4$-minor-free, it follows by \Cref{fact:cliqueminor} that for every $r,m \in \N$ there exists an $N=N(r,m) \in \N$ such that if $G$ is an outerplanar graph of size at least $N$, then there are disjoint sets $A,S \subseteq V(G)$ with $|A|\geq m, |S|\leq 2$, $A$ $r$-independent in $G\setminus S$, and $K_{A,S}$ complete. Since outerplanar graphs also forbid $K_{2,3}$ as a minor, this implies that $|S|\leq 1$. It is then clear that for any such $G$ and $S \subseteq V(G)$, the graph $\oplus_S ^n G$ is still outerplanar for all $n \in \N$, as no $1$-point free amalgams can create $K_4$ or $K_{2,3}$ minors. We thus conclude by \Cref{thm:main}.
 \end{proof}

Interestingly, the examples exhibited above are in fact \emph{almost-wide}, that is, the number of the bottleneck points does not depend on the radius of independence. It would be interesting to find natural quasi-wide classes which are not almost-wide, and which are closed under bottleneck amalgamation. One potential candidate might be the class of all graphs whose local treewidth is bounded by the same constant. 

\section{Homomorphism preservation fails on planar graphs}\label{sec:planar}

In this section we witness that homomorphism preservation fails on the class of planar graphs. Previously, it was established \cite{extensionspreservation} that the extension preservation property fails on planar graphs. Since extension preservation implies homomorphism preservation on hereditary classes by \Cref{cor:implies}, our result strengthens the above. Recall that by Wagner's theorem a graph is planar if and only if it omits $K_{3,3}$ and $K_5$ as minors. %This result is therefore of particular interest when juxtaposed with the fact that homomorphism preservation holds for the subclass of outerplanar graphs, as well as for the superclass of $K_{3,3}$-minor-free graphs. Moreover, as as it will be argued the same 
Our construction will in fact also reveal that homomorphism preservation fails on the class of $K_5$-minor-free graphs. 

\begin{definition}
Fix $n \in \N$. Define $G_n$ as the undirected graph on vertex set $V(G_n)=\{v_1,v_2\}\cup\{a_i : i \in [n]\} \cup \{b_i : i \in [n]\}$ and edge set
\[ E(G_n)=\{(v_1,a_i):i \in [n]\} \cup \{(v_2,b_i): i \in [n]\} \cup \{(a_i,b_i): i \in [n]\}\]\[\cup \{(a_i,a_{i+1}): i \in [n-1] \} \cup \{ (b_i,b_{i+1}): i \in [n-1]\} \cup \{(a_{i+1},b_i): i \in [n-1]\}.\]
We define $D_n$ as the extension of $G_n$ on the same vertex set, with 
\[ E(D_n) = E(G_n) \cup \{(a_1,a_n),(b_1,b_n),(a_1,b_n)\}.\]
In addition, we define $A_n$ as the graph obtained from $G_n$ by taking the quotient over the equivalence relation generated by $(a_1,a_n)$, and we write $\alpha_n:G_n \to A_n$ for the corresponding quotient homomorphism. Likewise, we let $B_n := G_n / (a_1,b_n)$ and $C_n:= G_n / (b_1,b_n)$, and write $\beta_n:G_n \to B_n$ and $\gamma_n:G_n \to C_n$ for the respective quotient homomorphisms. 
\end{definition}

\begin{figure}[h!]
    \centering
    \begin{tikzpicture}[scale=0.75]
  % Nodes

    \node[circle, draw, fill=black, inner sep=0.8pt,label=above:$v_1$] (X) at (0,2.2) {};
    \node[circle, draw, fill=black, inner sep=0.8pt,label=below:$v_2$] (Y) at (0,-2.2) {};
      \foreach \i in {1,...,9} {
    \node[circle, draw, fill=black, inner sep=0.8pt] (N\i) at (-5+\i,0.6) {};
  }
  \foreach \i in {1,...,9} {
    \node[circle, draw, fill=black, inner sep=0.8pt] (M\i) at (-5+\i,-0.6) {};
  }
    \node[circle, draw=none] (P) at (0,3) {};
    \node[circle, draw=none] (T) at (0,-3) {};
    \node[circle, draw=none] (W2) at (-0.5,-4) {};

  % Edges 
  
      \foreach \i in {1,...,9} {
    \draw (X.center) -- (N\i.center);
    \draw (Y.center) -- (M\i.center);
    \edef\next{\number\numexpr \ifnum \i<9 \i+1 \else \i\fi}
    \draw (N\i.center) -- (N\next.center);
    \draw (M\i.center) -- (M\next.center);
    \draw (M\i.center) -- (N\next.center);
    \draw (M\i.center) -- (N\i.center);
  }
\end{tikzpicture}
    \begin{tikzpicture}[scale=0.75]
  % Nodes

    \node[circle, draw, fill=black, inner sep=0.8pt,label=above:$v_1$] (X) at (0,2.2) {};
    \node[circle, draw, fill=black, inner sep=0.8pt,label=below:$v_2$] (Y) at (0,-2.2) {};
      \foreach \i in {1,...,9} {
    \node[circle, draw, fill=black, inner sep=0.8pt] (N\i) at (-5+\i,0.6) {};
  }
  \foreach \i in {1,...,9} {
    \node[circle, draw, fill=black, inner sep=0.8pt] (M\i) at (-5+\i,-0.6) {};
  }
    \node[circle, draw=none] (P) at (0,3) {};
    \node[circle, draw=none] (T) at (0,-3) {};
    \node[circle, draw=none] (W1) at (-5,-0.6) {};
    \node[circle, draw=none] (W2) at (-0.5,-4) {};

  % Edges 
  
      \foreach \i in {1,...,9} {
    \draw (X.center) -- (N\i.center);
    \draw (Y.center) -- (M\i.center);
    \edef\next{\number\numexpr \ifnum \i<9 \i+1 \else \i\fi}
    \draw (N\i.center) -- (N\next.center);
    \draw (M\i.center) -- (M\next.center);
    \draw (M\i.center) -- (N\next.center);
    \draw (M\i.center) -- (N\i.center);
  }
    \draw (N9.center) edge[out=110, in=0] (P.center);
    \draw (N1.center) edge[out=70, in=180] (P.center);
    \draw (M9.center) edge[out=250, in=0] (T.center);
    \draw (M1.center) edge[out=290, in=180] (T.center);
    \draw (M9.center) edge[out=250, in=0] (W2.center);
    \draw (W1.center) edge[out=270, in=180] (W2.center);
    \draw (W1.center) edge[out=90, in=180] (N1.center);
\end{tikzpicture}

\caption{Planar embeddings of $G_9$ and $D_9$ respectively.}
    \label{fig:d9}

\end{figure}
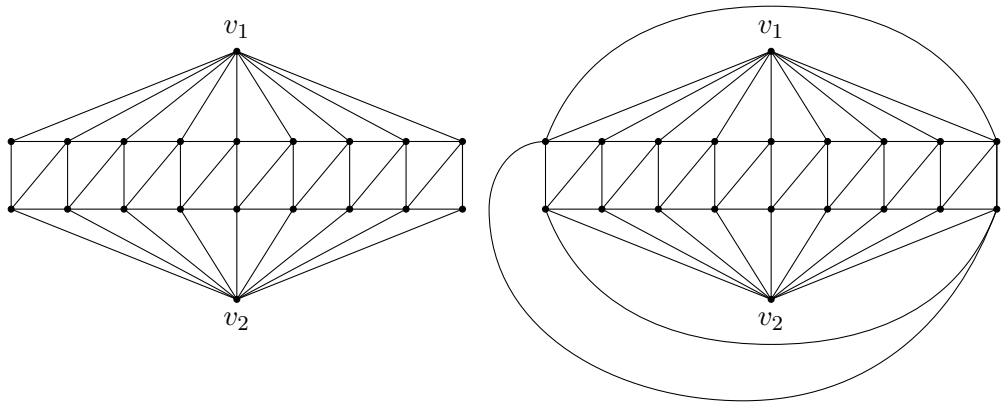

\begin{figure}[h!]
 
    \begin{tikzpicture}[scale=0.75]
  % Nodes

    \node[circle, draw, fill=black, inner sep=0.8pt] (X) at (0,2.2) {};
    \node[circle, draw, fill=black, inner sep=0.8pt] (Y) at (0,-2.2) {};
      \foreach \i in {1,...,5} {
    \node[circle, draw, fill=black, inner sep=0.8pt] (N\i) at (-3+\i,0.6) {};
  }
  \foreach \i in {1,...,5} {
    \node[circle, draw, fill=black, inner sep=0.8pt] (M\i) at (-3+\i,-0.6) {};
  }
    \node[circle, draw=none] (P) at (0,3) {};
    \node[circle, draw=none] (T) at (-0.5,-3) {};
    \node[circle, draw=none] (W1) at (-3.5,-0.6) {};
    \node[circle, draw=none] (W2) at (-0.75,-4) {};
    \node[circle, draw, fill=black, inner sep=0.8pt] (A) at (-3,-0.6) {};
  % Edges 
  
      \foreach \i in {1,...,5} {
    \draw (X.center) -- (N\i.center);
    \draw (Y.center) -- (M\i.center);
    \edef\next{\number\numexpr \ifnum \i<5 \i+1 \else \i\fi}
    \draw (N\i.center) -- (N\next.center);
    \draw (M\i.center) -- (M\next.center);
    \draw (M\i.center) -- (N\next.center);
    \draw (M\i.center) -- (N\i.center);
  }
    \draw (N5.center) edge[out=110, in=0] (P.center);
    \draw (N1.center) edge[out=70, in=180] (P.center);
    \draw (M5.center) edge[out=250, in=0] (T.center);
    \draw (A.center) edge[out=290, in=180] (T.center);
    \draw (M5.center) edge[out=250, in=0] (W2.center);
    \draw (W1.center) edge[out=270, in=180] (W2.center);
    \draw (W1.center) edge[out=80, in=190] (N1.center);
    \draw (A.center) -- (N1.center);
    \draw (A.center) -- (Y.center);
    
\end{tikzpicture}
\begin{tikzpicture}[scale=0.7]
  % Nodes

    \node[circle, draw, fill=black, inner sep=0.8pt] (X) at (0,2.2) {};
    \node[circle, draw, fill=black, inner sep=0.8pt] (Y) at (0,-2.2) {};
      \foreach \i in {1,...,5} {
    \node[circle, draw, fill=black, inner sep=0.8pt] (N\i) at (-3+\i,0.6) {};
  }
  \foreach \i in {1,...,5} {
    \node[circle, draw, fill=black, inner sep=0.8pt] (M\i) at (-3+\i,-0.6) {};
  }
    \node[circle, draw=none] (P) at (0.5,3) {};
    \node[circle, draw=none] (T) at (-2.5, -0.6) {};
    \node[circle, draw=none] (W1) at (-3,-0.6) {};
    \node[circle, draw=none] (W2) at (-0.5,-3) {};
    \node[circle, draw=none] (W3) at (-0.5,-4) {};
    \node[circle, draw=none] (W4) at (3.15,-4) {};
    \node[circle, draw, fill=black, inner sep=0.8pt] (A) at (3,0.6) {};
  % Edges 
  
      \foreach \i in {1,...,5} {
    \draw (X.center) -- (N\i.center);
    \draw (Y.center) -- (M\i.center);
    \edef\next{\number\numexpr \ifnum \i<5 \i+1 \else \i\fi}
    \draw (N\i.center) -- (N\next.center);
    \draw (M\i.center) -- (M\next.center);
    \draw (M\i.center) -- (N\next.center);
    \draw (M\i.center) -- (N\i.center);
    
  }

    \draw (N1.center) edge[out=70, in=180] (P.center);
    \draw (A.center) edge[out=110, in=0] (P.center);
    \draw (A.center) -- (M5.center);
    \draw (A.center) -- (N5.center);
    \draw (A.center) -- (X.center);
    \draw (Y.center) edge[out=180, in=270] (T.center);
    \draw (T.center) edge[out=90, in=200] (N1.center);
    \draw (M5.center) edge[out=270, in=0] (W2.center);
    \draw (W2.center) edge[out=180, in=270] (W1.center);
    \draw (W1.center) edge[out=90, in=180] (N1.center);
 
\end{tikzpicture}
\begin{tikzpicture}[scale=0.7]
  % Nodes

    \node[circle, draw, fill=black, inner sep=0.8pt] (X) at (0,2.2) {};
    \node[circle, draw, fill=black, inner sep=0.8pt] (Y) at (0,-2.2) {};
      \foreach \i in {1,...,5} {
    \node[circle, draw, fill=black, inner sep=0.8pt] (N\i) at (-3+\i,0.6) {};
  }
  \foreach \i in {1,...,5} {
    \node[circle, draw, fill=black, inner sep=0.8pt] (M\i) at (-3+\i,-0.6) {};
  }
    \node[circle, draw=none] (P) at (0,3) {};
    \node[circle, draw=none] (T) at (0,-3) {};
    \node[circle, draw=none] (W2) at (0.5,-4) {};
    \node[circle, draw, fill=black, inner sep=0.8pt] (A) at (3,0.6) {};
  % Edges 
  
      \foreach \i in {1,...,5} {
    \draw (X.center) -- (N\i.center);
    \draw (Y.center) -- (M\i.center);
    \edef\next{\number\numexpr \ifnum \i<5 \i+1 \else \i\fi}
    \draw (N\i.center) -- (N\next.center);
    \draw (M\i.center) -- (M\next.center);
    \draw (M\i.center) -- (N\next.center);
    \draw (M\i.center) -- (N\i.center);
  }

    \draw (M5.center) edge[out=250, in=0] (T.center);
    \draw (M1.center) edge[out=290, in=180] (T.center);
    \draw (M1.center) edge[out=270, in=180] (W2.center);
    \draw (A.center) edge[out=270, in=0] (W2.center);
    \draw (A.center) -- (M5.center);
    \draw (A.center) -- (N5.center);
    \draw (A.center) -- (X.center);
\end{tikzpicture}

\caption{Planar embeddings of $A_6, B_6,$ and $ C_6$ respectively.}
    \label{fig:abc}
\end{figure}
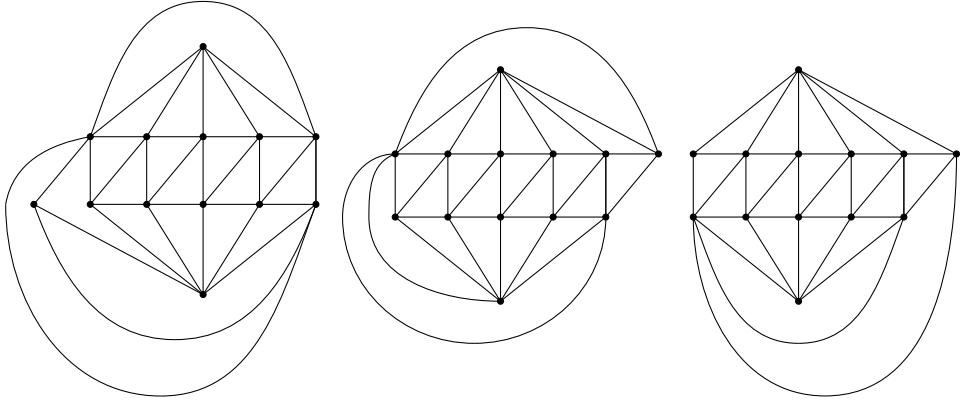

Consider the following observations. First, for every $n \geq 3$ the graphs $G_n,D_n,A_n,B_n,C_n$ are all planar and $4$-chromatic. In particular $D_n,B_n,C_n$ are maximal planar. %; in fact the addition of any other edge in these graphs produces a topological $K_5$-minor. 
For $n=3$ the graphs $D_n$ and $B_n$ contain a copy of $K_4$, while for $n \geq 4$ they are $K_4$-free. Likewise, for $n\in \{3,4\}$ the graphs $A_n$ and $C_n$ contain a copy of $K_4$, while for $n \geq 5$ they are $K_4$-free. Finally, for $3\leq m \leq n$ there is a homomorphism $\delta_{n,m}:G_n \to D_m$ that ``wraps'' $G_n$ around $D_m$.  Labelling their vertices as above, this satisfies
\[ \delta_{n,m}(v_1)=v_1, \delta_{n,m}(v_2)=v_2, \delta_{n,m}(a_i) = a_{i \bmod m} \text{ and } \delta_{n,m}(b_i) = b_{i \bmod m},\]
for all $i \in [n]$.

%We shall exhibit a formula whose minimal $K_5$-minor-free models will precisely be the graphs $D_n$ and $K_4$. Thus, our aim is to characterise the $K_5$-minor-free homomorphic images of $D_n$; in particular, we shall argue that these must necessarily contain $D_m$ as an induced subgraph for some $m \divides n$. Towards this, we argue by induction. Since the homomorphic images of $D_{n+1}$ are not a priori related to the images of $D_n$, 
We proceed to characterise the $K_5$-minor-free homomorphic images of $D_n$. We argue that these either contain $K_4$, or an induced copy of $D_m$ for some $m \divides n$. The proof proceeds by first characterising the $K_5$-minor-free homomorphic images of $G_n$ by induction, and then using that $G_n$ is a subgraph of $D_n$. While this requires a fair amount of book-keeping, it is not conceptually difficult. The base case of the induction is the following lemma. 

\begin{lemma}\label{lem:g3}
Let $f:G_3 \to H$ be a homomorphism. If $H$ is $K_4$-free then $f$ is injective. 
\end{lemma}

\begin{proof}
Let $H$ be a $K_4$-free graph, and $f:G_3 \to H$ a homomorphism. Label the vertices of $G_3$ as in the picture below; we shall argue that $f$ is injective.  

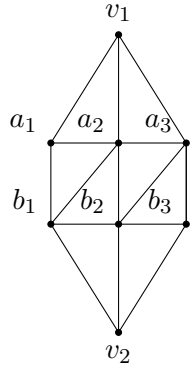
\begin{figure}[h!]
    \centering
    \begin{tikzpicture}[scale=0.9]
  % Nodes
    \node[circle, draw, fill=black, inner sep=0.8pt,label=above:$v_1$] (X) at (0,2.2) {};
    \node[circle, draw, fill=black, inner sep=0.8pt,label=below:$v_2$] (Y) at (0,-2.2) {};
      \foreach \i in {1,...,3} {
    \node[circle, draw, fill=black, inner sep=0.8pt,label=160:$a_\i$] (N\i) at (-2+\i,0.6) {};
  }
  \foreach \i in {1,...,3} {
    \node[circle, draw, fill=black, inner sep=0.8pt,label=160:$b_\i$] (M\i) at (-2+\i,-0.6) {};
  }
      \foreach \i in {1,...,3} {
    \draw (X.center) -- (N\i.center);
    \draw (Y.center) -- (M\i.center);
    \edef\next{\number\numexpr \ifnum \i<3 \i+1 \else \i\fi}
    \draw (N\i.center) -- (N\next.center);
    \draw (M\i.center) -- (M\next.center);
    \draw (M\i.center) -- (N\next.center);
    \draw (M\i.center) -- (N\i.center);
  }
\end{tikzpicture}
\end{figure}

First, notice that $f$ is injective on $\{v_2,b_1,b_2\}$ since they form a triangle in $G_3$. If $f(a_2)=f(v_2)$ then the set $\{f(v_2),f(b_2),f(a_3),f(b_3)\}$ induces $K_4$ in $H$; it follows that $f(a_2)\neq f(v_2)$, and so $f$ is injective on $\{v_2,b_1,b_2,a_2\}$. Likewise, $f(b_3)\neq  f(a_2)$ as otherwise $\{f(b_1),f(b_2),f(b_3),f(v_2)\}$ induce $K_4$ in $H$. Moreover, $f(b_3)\neq f(b_1)$ as $\{f(b_1),f(b_2),f(a_2),f(a_3)\}$ would otherwise induce $K_4$ in $H$. It follows that $f$ is injective on $\{v_2,b_1,b_2,b_3,a_2\}$. From this we deduce that $f(a_3)\neq f(b_1)$ as otherwise $\{f(v_2),f(b_1),f(b_2),f(b_3)\}$ would induce $K_4$ in $H$, and that $f(a_3)\neq f(v_2)$ as otherwise $\{f(v_2),f(b_1),f(b_2),f(a_2)\}$ would also induce $K_4$. Hence, $f$ is injective on $\{v_2,b_1,b_2,b_3,a_2,a_3\}$.
By symmetry, it follows that $f$ is also injective on $\{v_1,a_1,a_2,a_3,b_1,b_2\}$. Notice that $f(a_1)\neq f(b_3)$ and $f(v_1)\neq f(b_3)$ since otherwise $\{f(a_2),f(a_3),f(b_2),f(b_3)\}$ would induce $K_4$ in $H$. Finally, $f(a_1)\neq f(v_2)$ and $f(v_1)\neq f(v_2)$ as otherwise $\{f(a_2),f(b_1),f(b_2),f(v_2)\}$ would induce $K_4$ in $H$. Putting all the above together, we conclude that $f$ is injective on all of $G_3$ as required. 
\end{proof}

We now proceed to the general case. 

\begin{lemma}\label{prop:gninduction}
 Fix $n \geq 3$. Let $f:G_n \to H$ be a homomorphism, where $H$ is $K_4$-free and $K_5$-minor-free. Then one of the following is true:
    \begin{enumerate}
        \item\label{it:1} $f$ is injective;
        \item\label{it:2} there is some $m \in [4,n-1]$ and an embedding $\hat f: D_m \to H$ such that $f = \hat f \circ \delta_{n,m}$;
        \item\label{it:3} $n \geq 5$ and there is an injective homomorphism $\hat f: A_n \to H$ such that $f = \hat f \circ \alpha_n$;
        \item\label{it:4} $n \geq 4$ and there is an embedding $\hat f: B_n \to H$ such that $f = \hat f \circ \beta_n$;
        \item\label{it:5} $n \geq 5$ and there is an embedding $\hat f: C_n \to H$ such that $f = \hat f \circ \gamma_n$;

    \end{enumerate}
\end{lemma}

\begin{proof}
    We prove the claim by induction on $n$. The base case $n = 3$ follows by \Cref{lem:g3}. So, fix a $K_4$-free and $K_5$-minor-free graph $H$ and consider a homomorphism $f:G_{n+1} \to H$. Evidently, this restricts to a homomorphism $f': G_n \to H$. By the induction hypothesis, we may assume that $f'$ satisfies one of the five conditions of this proposition. 
    
    Assume that $f'$ satisfies (\ref{it:1}), i.e. $f$ is injective on $G_n= G_{n+1} \setminus \{a_{n+1},b_{n+1}\}$. We consider the images of the vertices $a_{n+1}$ and $b_{n+1}$ under $f$. Observe that $(a_{n+1},b_{n+1}) \in E(G_{n+1})$ so $f(a_{n+1})\neq f(b_{n+1})$. Clearly, if $f(a_{n+1})$ and $f(b_{n+1})$ are not any of the vertices in $f[G_n]$, then $f$ is itself injective and (\ref{it:1}) holds. We hence distinguish three cases. 

    First, suppose that ${f(a_{n+1}) \in f[G_n]}$ and ${f(b_{n+1}) \notin f[G_n]}$. Since there are edges $(v_1,a_{n+1}),(a_n,a_{n+1}),$ and $(b_n,a_{n+1})$, it follows that $f(a_{n+1}) \notin \{f(v_1),f(a_n),f(b_n)\}$. Moreover, $f(a_{n+1}) \neq f(v_2)$ as otherwise $\{f(b_{n-1}),f(b_n),f(a_n),f(v_2)\}$ would induce $K_4$ in $H$. Similarly, $f(a_{n+1})\neq f(a_{n-1})$, as otherwise $\{f(a_{n-1}),f(b_{n-1}),f(a_n),f(b_n)\}$ would induce $K_4$, and $f(a_{n+1})\neq f(b_{n-1})$, as otherwise $\{f(v_1),f(a_{n-1}),f(a_n),f(b_{n-1})\}$ would induce $K_4$. In addition, $f(a_{n+1}) \notin \{f(a_i): 2 \leq i \leq n-2\}$ as otherwise an edge $(f(a_i),f(a_n))$ for some $i \in [2,n-2]$ would produce a $K_5$-minor in $H$, namely the minor arising from $S_1=\{f(v_1)\},S_2=\{f(a_i)\},S_3=\{f(a_{n})\},S_4=\{f(a_j): i+1 \leq j \leq n-1\}, S_5=\{f(v_2),f(b_1),f(a_1),f(b_i),f(b_{i+1}),f(b_n)\}$. A similar argument reveals that $f(a_{n+1}) \notin \{f(b_i): 2 \leq i \leq n-2\}$. It follows that $f(a_{n+1})=f(b_1)$ or $f(a_{n+1})=f(a_1)$. The former case would produce a copy of $K_4$, namely $\{f(v_1),f(a_1),f(a_2),f(b_1)\}$, leading to a contradiction. Hence $f(a_{n+1}) = f(a_1)$. Since $f(b_{n+1}) \notin f[G_n]$, it follows that $f$ factors through the quotient homomorphism $\alpha_n$, i.e. case (\ref{it:3}) is true. 

    Next, suppose that $f(a_{n+1}) \in f[G_n]$ and $f(b_{n+1}) \in f[G_n]$. As before, we deduce from the first assumption that $f(a_{n+1})=f(a_1)$. This implies that there are edges $(f(a_1),f(a_n))$ and $ (f(a_1),f(b_n))$ in $H$. Since there are edges $(a_{n+1},b_{n+1}),(b_n,b_{n+1}),(v_2,b_{n+1})$ in $G_{n+1}$ we deduce that $f(b_{n+1})\notin \{f(a_1),f(b_n),f(v_2)\}$. Moreover, $f(b_{n+1})\neq f(v_1)$ as otherwise the edge $(f(v_1),f(v_2))$ would produce a $K_5$-minor in $H$, namely $S_1=\{f(v_1)\},S_2=\{f(a_1)\},S_3=\{f(a_n)\}, S_4=\{f(a_j):j \in [2,n-1]\}, S_5=\{f(v_2),f(b_1),f(b_2),f(b_n)\}$. Similarly, a $K_5$-minor arises in $H$ if $f(b_{n+1}) \in \{f(a_i),f(b_i): i \in [2,n-1]\}$. We thus deduce that $f(b_{n+1})=f(b_1)$, from which we conclude that $f[G_{n+1}]$ induces a copy of $D_n$ in $H$, and more precisely, case (\ref{it:2}) holds. 

    So, suppose that $f(a_{n+1}) \notin f[G_n]$ and $f(b_{n+1}) \in f[G_n]$. We consider the possible images of $b_{n+1}$ under $f$. Since there are edges $(a_{n+1},b_{n+1}),(v_2,b_{n+1}),(b_n,b_{n+1})$ in $G_{n+1}$ we deduce that $f(b_{n+1})\notin \{f(a_{n+1}),f(v_2),f(b_n)\}$. Moreover, $f(b_{n+1}) \neq f(v_1)$ as otherwise $\{f(v_1),f(a_n),f(a_{n+1}),f(b_n)\}$ would induce $K_4$ in $H$. Likewise, $f(b_{n+1})\neq f(a_n)$ as otherwise $\{f(v_2),f(a_{b-1}),f(b_n),f(a_n)\}$ would induce $K_4$ in $H$. Moreover $f(b_{n+1}) \notin \{f(a_i):i \in [2,n-1]\}$ as otherwise the edge $(f(a_i),f(a_{n+1}))$ would produce a $K_5$-minor in $H$, namely $S_1=\{f(v_1)\},S_2=\{f(a_i)\},S_3=\{f(a_{n+1})\}, S_4 = \{f(a_j): j \in [i+1,n]\},S_5=\{f(v_2),f(b_1),f(a_1),f(b_i),f(b_{i+1}),f(b_n)\}$. Very similarly, we deduce that $f(b_{n+1}) \notin \{f(b_i):i \in [2,n-1]\}$. It follows that $f(b_{n+1})=f(a_1)$ or $f(b_{n+1})=f(b_1)$. In the former case, it follows that $f$ factors through the quotient homomorphism $\beta_n$, and so (\ref{it:4}) is true, while in the latter case, it follows that $f$ factors through $\gamma_n$, and so (\ref{it:5}) is true. 

    Next, assume that $f'$ satisfies (\ref{it:2}), i.e. there is some $m \in [4,n-1]$ and and an embedding $\hat f:D_m \to H$ such that $f' = \hat f \circ \delta_{n,m}$. Arguing as before, it is easy to see that the assumptions on $H$ force $f(a_{n+1})$ to be equal to $\hat f(a_{n+1\bmod m} )$, and likewise $f(b_{n+1})=\hat f(b_{n+1\bmod m})$, implying that $f = \hat f \circ \delta_{n+1,m}$. Hence $f$ also satisfies (\ref{it:2}). 

    Finally, we argue that $f'$ cannot satisfy any of $(3),(4),$ and (\ref{it:5}). Indeed, assume for a contradiction that $f'$ satisfies (\ref{it:3}) and write $f'$ as $\hat f \circ \alpha_n$ for some injective homomorphism $\hat f:A_n \to H$. In particular, we know that $n \geq 5$. Consider the image of $a_{n+1}$ under $f$; this is some vertex adjacent to $f(a_n)=f(a_1), f(v_1),$ and $f(b_n)$. If $f(a_{n+1})\notin f[G_n]$, then we obtain a $K_5$-minor in $H$, namely $S_1=\{f(v_1)\}, S_2=\{f(a_1)\},S_3=\{f(a_{n-1})\}, S_4=\{f(a_i): i \in [2,n-1]\}, S_5=\{f(v_2),f(b_1),f(b_2),f(b_{n-1}),f(b_n),f(a_{n+1})\}$. So $f(a_{n+1})\in f[G_n]$. If $f(a_{n_1})=f(a_i)$ for some $i \in [2,n-2]$ then we obtain a $K_5$-minor in $H$ by picking some $j \in [2,n-2]\setminus\{i\}$ and letting $S_1=\{f(v_1),f(a_j)\},S_2=\{f(a_1)\},S_3=\{f(a_{n-1})\},S_4=\{f(b_i):i \in [n-1]\}, S_5=\{f(v_2),f(b_n),f(a_i)\}$. Likewise, if $f(a_{n+1})= f(a_{n-1})$ then we obtain the $K_5$-minor $S_1 = \{f(v_1)\},S_2=\{f(a_1)\},S_3=\{f(b_{n-1}),f(b_{n-2}),f(a_{n-2})\},S_4=\{f(a_i): i \in [2,n-3]\},S_5=\{f(v_2),f(b_n),f(a_{n-1})\}$. Consequently, $f(a_{n+1})=f(b_i)$ for some $i \in [1,n-1]$. This produces an edge $(f(v_1),f(b_i))$ and thus gives rise to the $K_5$-minor $S_1=\{f(v_1)\},S_2=\{f(a_j):j \in [1,i-1]\}, S_3=\{f(a_i)\},S_4=\{f(a_j):j \in [i+1,n-1]\},S_5=\{f(v_2),f(b_1),f(b_i),f(b_{n-1})\}$. It follows that $f'$ cannot satisfy (\ref{it:3}); via very similar reasoning, we exclude cases (\ref{it:4}) and (\ref{it:5}). 
\end{proof}

Having established the above, our characterisation of the $K_5$-minor-free homomorphic images of $D_n$ follows easily. 

\begin{restatable}{proposition}{plnrimage}\label{prop:planarimage}
    Fix $n \geq 4$. Then any $K_4$-free and $K_5$-minor-free homomorphic image of $D_n$ contains an induced copy of $D_m$ for some $m\geq 4$ such that $m \divides n$.    
\end{restatable}

\begin{proof}
    Consider a homomorphism $f:D_n \to H$ where $H$ is $K_4$-free and $K_5$-minor-free. Then $f$ descends to a homomorphism $f': G_n \to H$. It follows that one of the five cases of \Cref{prop:gninduction} holds. If $f'$ is injective, then in particular $f$ is injective; since the addition of any edge in $D_n$ creates a $K_5$-minor, it follows that $f$ is in fact an embedding as required. Suppose that case (\ref{it:2}) is true, and let $m \in [4,n-1]$ and $\hat f:D_m \to H$ be such that $f'=\hat f \circ \delta_{n,m}$. In particular, $D_m$ is an induced subgraph of $H$ and $m \divides n$. Finally, case (\ref{it:3}) leads to a contradiction as the edge $(a_1,a_n)$ in $D_n$ implies that $f(a_1)\neq f(a_n)$, case (\ref{it:4}) leads to a contradiction as the edge $(a_1,b_n)$ implies that $f(a_1)\neq f(b_n)$, and likewise, case (\ref{it:5}) leads to a contradiction as the edge $(b_1,b_n)$ implies that $f(b_1)\neq f(b_n)$. 
\end{proof}

We also define $G_\infty$ as the countably infinite analogue of $G_n$, i.e. the graph on the vertex set $V(G_\infty) = \{v_1,v_2\}\cup \{a_i: i \in \N_{>0}\}\cup \{b_i: i \in \N_{>0}\}$ and edge set 
\[ E(G_n)=\{(v_1,a_i):i \in \N_{>0}\} \cup \{(v_2,b_i): i \in \N_{>0}\} \cup \{(a_i,b_i): i \in \N_{>0}\}\]\[\cup \{(a_i,a_{i+1}): i \in \N_{>0} \} \cup \{ (b_i,b_{i+1}): i \in \N_{>0}\} \cup \{(a_{i+1},b_i): i \in \N_{>0}\}.\]
Likewise, we define the homomorphism $\delta_{\infty,m}:G_\infty \to D_m$ in analogy to the homomorphisms $\delta_{n,m}:G_n \to D_m$. \Cref{prop:gninduction} allows us to also characterise the finite $K_5$-minor-free homomorphic images of $G_\infty$. 

\begin{restatable}{lemma}{corinfty}\label{cor:infty}
    Let $f:G_\infty \to H$ be a homomorphism where $H$ is finite, $K_4$-free, and $K_5$-minor-free. Then there is some $m \geq 4$ and an embedding $\hat f: D_m \to H$ such that $f = \hat f \circ \delta_{\infty,m}$. 
\end{restatable}

\begin{proof}
    Fix $f$ as above and let $n:=|H|$. It follows that $f$ restricts to a homomorphism $f':G_n \to H$; this satisfies one of the five cases of \Cref{prop:gninduction}. Clearly, cases (\ref{it:1}),(\ref{it:3}),(\ref{it:4}), and (\ref{it:5}) are ruled out due to size restrictions. Consequently, case (\ref{it:2}) holds and the claim follows. 
 \end{proof}

We proceed to show that the existence of the graphs $D_n$ as induced subgraphs is definable among $K_4$-free $K_5$-minor-free graphs by a simple first-order formula. Indeed, consider the formulas
\[ \chi(x_1,x_2,y_1,z_1,y_2,z_2) =  E(x_1,y_2) \land E(y_1,y_2) \land E(z_1,y_2) \land E(z_1,z_2) \land E(y_2,z_2) \land  E(z_2,x_2),\text{ and}\]
\[ \phi = \exists x_1,x_2,y,z [ E(x_1,y)\land E(y,z)\land E(z,x_2)  \land \forall a,b (E(x_1,a)\land E(a,b)\land E(b,x_2))\]\[ \to \exists c,d \ \chi(x_1,x_2,a,b,c,d))]\]

\begin{lemma}\label{prop:planarinduced}
Let $H$ be a finite $K_4$-free and $K_5$-minor free graph. If $H \models \phi$ then there is some $n \geq 4$ such that $H$ contains $D_n$ as an induced subgraph. 
\end{lemma}

\begin{proof}
    Fix a graph $H$ as above, and suppose that $G\models \phi$. We inductively define a chain of partial homomorphisms $f_1 \subseteq f_2 \subseteq f_3 \subseteq \dots$ from $G_\infty \to H$ such that $\mathrm{dom}(f_n)=G_n$. Then the map $f = \cup_{n = 1}^\infty f_n$ is a homomorphism $G_\infty \to H$, and hence \Cref{cor:infty} implies that $H$ contains some $D_n$ as an induced subgraph. 
    
    Since $H \models \phi$ it follows that there are $x_1,x_2,y,z \in V(H)$ such that 
    \[H \models E(x_1,y)\land E(y,z)\land E(z,x_2).\]
    Consequently, the map $f_1:G_1 \to H$ given by $f(v_1) = x_1, f(v_2)= x_2, f(a_1) = y, f(b_1)=z$ is a homomorphism as required. So, suppose that $f_n$ has been defined. Since 
    \[ H \models E(x_1,f(a_n))\land E(f(a_n),f(b_n)) \land E(f(b_n),x_2) \]
    it follows that 
    \[ H \models \exists c,d \ \chi(x_1,x_2,f(a_n),f(b_n),c,d).\]
    We consequently extend $f_n:G_n \to H$ to $f_{n+1}:G_{n+1} \to H$ by letting $f_{n+1}(a_{n+1})=c$ and $f_{n+1}(b_{n+1})=d$; this is easily seen to be a valid homomorphism by the choice of $\chi$. 
\end{proof}

\begin{lemma}\label{prop:planarmodels}
    Let $H$ be a $K_5$-minor-free graph. If $H$ contains some $D_n$ for $n \geq 3$ as an induced subgraph then $H \models \phi$. 
\end{lemma}

\begin{proof}
    Let $D_n \leq H$ be as above. Clearly, 
    \[ H \models E(v_1,a_1)\land E(a_1,b_1) \land E(b_1,v_2).\]
    So, let $a,b \in V(H)$ be arbitrary vertices such that
    \[ H \models E(v_1,a)\land E(a,b) \land E(b,v_2);\]
    we first argue that $a \in \{a_i :i \in [n]\}$ and $b \in \{b_i : i \in [n]\}$. Towards this, observe first that $b \notin \{a_i :i \in [n]\}\cup\{v_1\}$ as otherwise the edge $(v_2,a_i)$ or $(v_2,v_1)$ would contradict that $D_n$ is induced in $H$. So, assume for a contradiction that $a \notin \{a_i :i \in [n]\}$. Since there is an edge $(v_1,a)$ it follows that $a \neq v_1$, and so in particular the sets $S_1=\{v_1\},S_2=\{a_1\},S_3=\{a_n\},S_4=\{a_j : j \in [2,n-1]\},S_5=\{v_2,b_1,b_n,a,b\}$ produce a $K_5$-minor in $H$. With a symmetric argument we obtain that $b \in \{b_i : i \in [n]\}$. 
    Now, since there is an edge $(a,b)$ it follows that there is some $i \in [n]$ such that $a = a_i$ and $b= b_i$ or $b= b_{i - 1 \bmod n}$. In the former case we have that 
    \[H \models \chi(v_1,v_2,a,b,a_{i+1 \bmod n},b_{i+1  \bmod n}), \]
    while in the latter we have that 
    \[ H \models \chi(v_1,v_2,a,b,a_{i-1 \bmod n},b_{i-2 \bmod n}). \]
    In either case, 
    \[ H \models \exists c,d \ \chi(v_1,v_2,a,b,c,d),\]
    and since the choice of $a,b \in V(H)$ was arbitrary we obtain that $H \models \phi$ as required. 
\end{proof}

Putting all the above together, we deduce the main theorem of this section. 

\begin{theorem}
The class of planar graphs does not have the homomorphism preservation property. 
\end{theorem}

\begin{proof}
    Let $\hat{\phi}$ be the disjunction of $\phi$ with the formula that induces a copy of $K_4$, i.e. 
    \[\hat{\phi} := \phi \lor \exists x_1, x_2, x_3, x_4 \bigwedge_{i\neq j} E(x_i,x_j).\]
     We argue that $\hat{\phi}$ is preserved by homomorphisms over the class of planar graphs. Indeed, let $f:G \to H$ be a homomorphism with $G,H$ planar such that $G \models \hat\phi$. Clearly, if $H$ contains a copy of $K_4$ then $H \models \hat\phi$. So, without loss of generality we may assume that $G \models \phi$ and $G,H$ are $K_4$-free. It follows by \Cref{prop:planarinduced} that there exists some $n\geq 4$ such that $G$ contains $D_n$ as a subgraph. Consequently, \Cref{prop:planarimage} implies that there is some $m \geq 4$ such that $H$ contains $D_m$ as a subgraph. \Cref{prop:planarmodels} then implies that $H \models \phi$, and thus $H \models \hat\phi$ as required. To conclude, observe that the minimal models of $\hat\phi$ over the class of planar graphs are $K_4$ and the graphs $D_n$ for $n \geq 4$; since these are infinitely many \Cref{lem:minimalmodels} implies that $\hat\phi$ is not equivalent to an existential-positive formula over the class of planar graphs. 
\end{proof}

Since in the above argument we only use exclusion of $K_5$-minors, the same proof relativises to the following theorem. 

\begin{theorem}
The class of all $K_5$-minor-free graphs does not have the homomorphism preservation property. 
\end{theorem}

Finally, while we have not referred to topological minors to simplify our arguments, an easy check reveals that the above are still valid when considering graphs that forbid $K_5$ as a topological minor, implying that homomorphism preservation also fails on the class of $K_5$-topological-minor-free graphs.

\section{Conclusion}

Much work in finite model theory explores \emph{tame} classes of finite structures.  In~\cite{dawar2007finite}, two related notions of tameness are identified: algorithmic tameness and model-theoretic tameness.  The former is centred around the tractability of model-checking for first-order logic while the latter is illustrated by preservation theorems and it was argued that these occurred together in sparse classes of structures.  More recently, algorithmic tameness has been explored extensively for dense classes as well (see~\cite{dreier2023firstorder} for example).  On the other hand, the results here show that the status of preservation theorems on sparse classes is more subtle and relies on closure properties that are not always present in natural classes such as planar graphs.  Nonetheless, it is an interesting question whether the recent understanding of tame dense classes, such as \emph{monadically stable} and more generally \emph{monadically dependent} classes, can also cast light on preservation theorems.  
The arguments for homomorphism and extension preservation rely on quasi-wideness and almost-wideness respectively.  Similar wideness phenomena occur for dense classes, by replacing deletion of bottleneck points with performing \emph{flips}, that is, edge-complementations within subsets of the domain (see Table 1 in \cite{flipbreak}).  A question inspired by this is whether, for every $k \in \N$, the class of all graphs of cliquewidth at most $k$ has the extension preservation property. For $k\leq 2$ this follows from the fact that cographs are well-quasi-ordered by the induced subgraph relation \cite{damaschke1990induced}.

%\begin{question}
%Let $\Fcal$ be a finite set of finite graphs. What are the necessary and sufficient conditions so that $\mathrm{Forb}^\mathrm{min}(\Fcal)$ has the homomorphism preservation property?
%\end{question}

%Fix $k \in \N$. Does the class of all graphs of cliquewidth $\leq k$ have the extension preservation property?

%Fix $k,\ell \in \N$. Does the class of all graphs of cliquewidth $\leq k$ and ladder index $\leq \ell$ have the extension preservation property?

%Fix a function $f:\N \to \N$. Does the class of all graphs whose local treewidth is bounded by $f$ have the homomorphism preservation property?

%\begin{conjecture}
   % For every $k \geq 5$, the class of all $K_k$-minor-free graphs does not have the homomorphism preservation property. 
%\end{conjecture}

\bibliographystyle{plain}
\bibliography{bibliography}

\end{document}